\documentclass[a4paper,UKenglish,cleveref, autoref, numberwithinsect]{lipics-v2021}
\usepackage{xparse}
\usepackage{xargs}
\usepackage[linesnumbered,lined,boxed,commentsnumbered,noend]{algorithm2e}
\usepackage{algpseudocode}
\usepackage{booktabs}

\newcommand{\floor}[1]{\left\lfloor #1 \right\rfloor}
\newcommand{\ceil}[1]{\left\lceil #1 \right\rceil}
\newcommand{\eps}{\varepsilon}
\newcommand{\Oh}{\mathcal{O}}
\newcommand{\Ot}{\widetilde{\Oh}}
\newcommand{\Tht}{\widetilde{\Theta}}
\newcommand{\nat}{\mathbb{N}}
\newcommand{\Z}{\mathbb{Z}}
\newcommand{\Ss}{\mathcal{S}}
\newcommand{\Cc}{\mathcal{C}}
\newcommand{\Pp}{\mathcal{P}}

\newcommand{\poly}{\mathrm{poly}}
\newcommand{\norm}[1]{\lVert #1 \rVert}
\newcommand{\polylog}{\,\mathrm{polylog}}
\newcommand{\ol}{\overline}
\newcommand{\up}{\uparrow}
\newcommand{\dn}{\downarrow}

\renewcommand{\le}{\leqslant}
\renewcommand{\ge}{\geqslant}

\title{Knapsack and Subset Sum with Small Items}

\author{Adam Polak}{EPFL, Lausanne, Switzerland \and \url{https://adampolak.github.io/}}{adam.polak@epfl.ch}{https://orcid.org/0000-0003-4925-774X}{%
Supported by the Swiss National Science Foundation within the project \emph{Lattice Algorithms and Integer Programming} (185030).
Part of this work was done at Jagiellonian University, supported by Polish National Science Center grant 2017/27/N/ST6/01334}

\author{Lars Rohwedder}{EPFL, Lausanne, Switzerland \and \url{https://larsrohwedder.com/}}{lars.rohwedder@epfl.ch}{https://orcid.org/0000-0002-9434-4589}{Swiss National Science Foundation project 200021-184656}

\author{Karol W\k{e}grzycki}{Saarland University and Max Planck Institute for
Informatics}{wegrzycki@cs.uni-saarland.de}{https://orcid.org/0000-0001-9746-5733}{Project TIPEA that has
    received funding from the European Research Council (ERC) under the European Unions Horizon
2020 research and innovation programme (grant agreement No. 850979).}

\authorrunning{A.~Polak, L.~Rohwedder, and K.~W\k{e}grzycki}

\Copyright{Adam Polak, Lars Rohwedder, and Karol W\k{e}grzycki}



\nolinenumbers 

\hideLIPIcs  

\begin{document}

\maketitle

\begin{picture}(0,0)
\put(340,-370)
{\hbox{\includegraphics[width=40px]{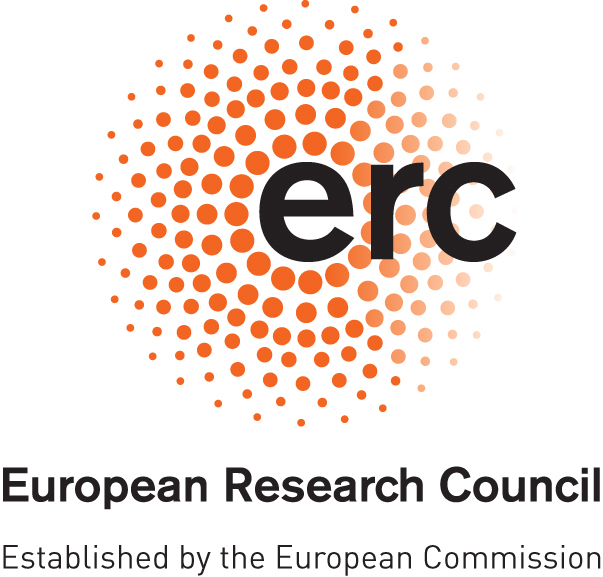}}}
\put(290,-390)
{\hbox{\includegraphics[width=60px]{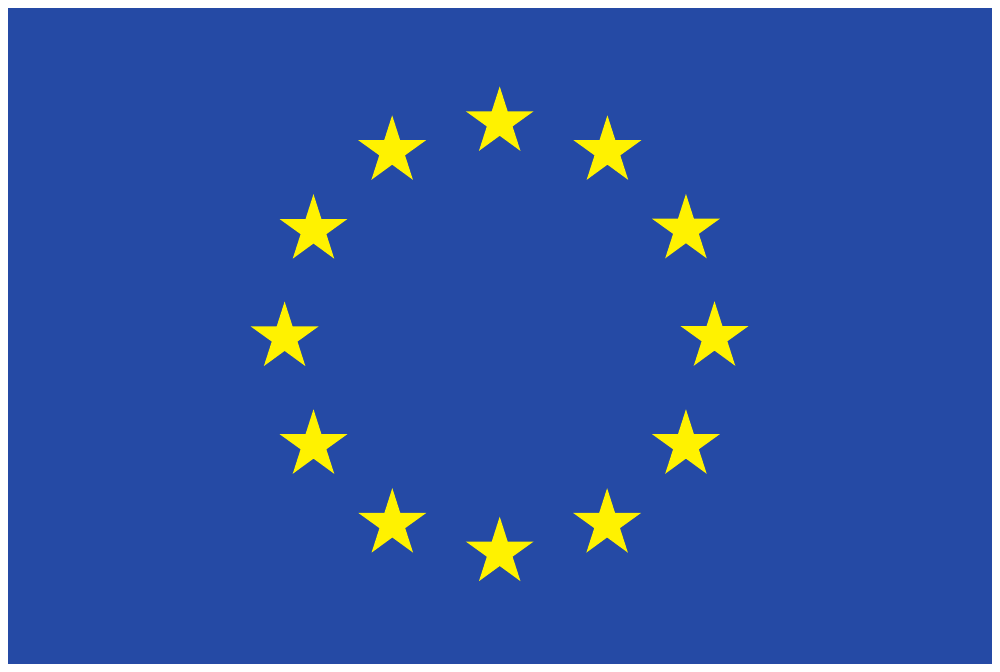}}}
\end{picture}

\begin{abstract}
Knapsack and Subset Sum are fundamental NP-hard problems in combinatorial optimization. Recently there has been a growing interest in understanding the best possible pseudopolynomial running times for these problems with respect to various parameters.

In this paper we focus on the maximum item size $s$ and the maximum item value $v$. 
We give algorithms that run in time $\Oh(n + s^3)$ and $\Oh(n + v^3)$ for the Knapsack problem, and in time $\Ot(n + s^{5/3})$ for the Subset Sum problem.

Our algorithms work for the more general problem variants with multiplicities, where each input item comes with a (binary encoded) multiplicity, which succinctly describes how many times the item appears in the instance. In these variants $n$ denotes the (possibly much smaller) number of \emph{distinct} items.

Our results follow from combining and optimizing several diverse lines of research, notably
proximity arguments for integer programming due to Eisenbrand and Weismantel (TALG 2019), 
fast structured $(\min,+)$-convolution by Kellerer and Pferschy (J.~Comb.~Optim.~2004),
and additive combinatorics methods originating from Galil and Margalit (SICOMP 1991).
\end{abstract}

\section{Introduction}

In the Knapsack problem we are given a (multi-)set consisting of $N$ items, where the $i$-th item has size $s_i$ and value $v_i$, and a knapsack capacity $t$. The task is to find a subset of items with the maximum total value such that its total size does not exceed the capacity $t$.
In the related Subset Sum problem we are given a (multi-)set of $N$ positive integers and a target value $t$, and the task is to find a subset of integers with the total sum exactly equal to $t$. The Subset Sum problem can thus be seen as a decision variant of the Knapsack problem with the additional restriction that $s_i = v_i$ for every item $i$.

Knapsack and Subset Sum are fundamental problems in computer science and
discrete optimization. They are studied extensively both from practical and
theoretical points of view (see, e.g.~\cite{kellerer-book} for a comprehensive
monograph). The two problems are (weakly) NP-hard,
and Bellman's seminal work on dynamic programming~\cite{bellman} gives
pseudopolynomial $\Oh(Nt)$ time algorithms for both of them. Recently there has been a growing interest in understanding the best possible pseudopolynomial running times for these problems with respect to various parameters, see, e.g., \cite{Bringmann17,bateni18,EisenbrandW18,AxiotisT19,dense-subset-sum}.

In this paper we consider binary multiplicity encoding of the Knapsack and Subset
Sum instances.  Each item $i$ given in the input has a (binary encoded) positive
integer multiplicity $u_i$, which denotes that up to $u_i$ copies of this item
can be used in a solution. In these variants $n$ denotes the (possibly much
smaller) number of distinct items and $N = \sum_{i \in [n]} u_i$.\footnote{We
use $[n]$ to denote $\{1,2,\ldots,n\}$.} 
Binary multiplicity encoding can be challenging because one requires the algorithm
to run in polynomial time in the input size, which can be exponentially smaller compared
to the naive encoding. A notable example of this setting is the breakthrough result of 
Goemans~and~Rothvo\ss~\cite{goemans-rothvoss14} showing that Bin Packing with few different item sizes
(and binary multiplicity encoding) can be solved in polynomial time.

Formally, Knapsack with multiplicities can be defined as an integer linear
program: maximize $\sum_{i \in [n]} v_ix_i$ subject to $0 \le x_i \le u_i$,
$x_i\in\Z$, and $\sum_{i\in [n]}s_ix_i \le t$. Similarly, Subset Sum with
multiplicities can be defined as a feasibility integer linear program with
constraints $0 \le x_i \le u_i$, $x_i\in\Z$, and $\sum_{i\in [n]}s_ix_i = t$.
Throughout the paper we use $u$ to denote the maximum item multiplicity
$\max_{i\in [n]} u_i$, and w.l.o.g.~we assume that $u \le t$.

\subsection{Our results}

We focus on pseudo-polynomial time algorithms with respect to the
maximum item size $s=\max_{i\in[n]}s_i$
(or maximum item value $v=\max_{i\in[n]}v_i$, which is essentially equivalent).
We note that $s$ is a stronger parameter compared to $t$ in the sense that $s$ can be much smaller than $t$,
but not vice versa. Yet, $s$ is less well understood than $t$.
In the regime where $n$ is large compared to $s$,
an $\Oh(n + \poly(s))$ time algorithm would be desirable.
We show that the Knapsack problem can indeed be solved
in such a time. Prior results (even for 0-1 Knapsack, that is, without multiplicities)
only came with the form $\Oh(\poly(n) \cdot \poly(s))$ or $\Oh(\poly(s) \cdot \poly(t))$.

This raises the natural question what the best exponent in the polynomial is.
In this paper we address the question from the upper bound side. We give algorithms for Knapsack running in $\Oh(n + s^3)$ and $\Oh(n + v^3)$ time, and for Subset Sum running in $\Ot(n + s^{5/3})$ time\footnote{Throughout the paper we use a $\Ot(\cdot)$ notation to hide
polylogarithmic factors.}.
Our algorithms are in the word RAM model, and we assume that each integer in the
input fits in one word. In particular, arithmetic operations on integers of size
polynomial in the sum of the input's integers require constant time.

Our first result is an algorithm for Knapsack.
We use proximity techniques due to Eisenbrand and Weismantel~\cite{EisenbrandW18} which
allow us to prove that there is an efficiently computable solution that differs only
very little from an optimal solution.
Then we apply a fast algorithm for structured $(\min,+)$-convolution~\cite{kellerer04} to search for this optimal solution
within the limited space. This results in a running time which is cubic in the maximum item size.
\begin{theorem}
\label{thm:knapsack}
    Knapsack (with multiplicities) can be solved in deterministic $\Oh(n+s^3)$ time.
\end{theorem}
The definition of (the decision variant of) the Knapsack problem is symmetric with respect to sizes and values.
We give a simple transformation that allows us to apply our algorithm also to the case where
the maximal item value (and not the maximal item size) is small.

\begin{theorem}
    \label{thm:knapsack-small-values}
    Knapsack (with multiplicities) can be solved in deterministic $\Oh(n+v^3)$ time.
\end{theorem}
Theorem~\ref{thm:knapsack} already implies that Subset Sum can also be solved in
$\Oh(n+s^3)$. Our algorithm uses as a subprocedure an $\Oh(n+st)$ time Knapsack algorithm (see Lemma~\ref{lem:maxplus}). If we simply replaced it with a $\Ot(N+t)$ time algorithm for Subset Sum~\cite{Bringmann17}, we would get a $\Ot(n+s^2)$ time algorithm for Subset Sum.
We improve on this by introducing a refined proximity argument that lets us further reduce an instance, where the maximum item multiplicity $u$ can be of the order of $s$, to two instances with $u \ll s$ each.
By combining this with additive combinatorics methods, originally developed by Galil and Margalit~\cite{galil-margalit} and recently generalized by Bringmann and Wellnitz~\cite{dense-subset-sum}, we then obtain a subquadratic algorithm.
\begin{theorem}
\label{thm:subsetsum}
    Subset Sum (with multiplicities) can be solved in randomized $\Ot(n+s^{5/3})$ time, with a one-sided error algorithm that returns a correct answer with high probability.
\end{theorem}

All our algorithms can also retrieve a solution, without increasing the
asymptotic running times. 
This is notable especially for our Subset Sum algorithm, in which we use as a
black box the Bringmann-Wellnitz algorithm that gives only yes/no answers. We can deal with
this presumable obstacle because we can afford to spend
more time on retrieving a solution than the Bringmann-Wellnitz algorithm could.

A limitation of our algorithms is that they can provide an answer only for a single target value $t$ at a time. Conversely, many (but not all) known Knapsack and Subset Sum algorithms can give answers for all target values between $0$ and $t$ at the same time. This limitation is however unavoidable: We aim at running times independent of the target value $t$, thus we cannot afford output size linear in $t$, because $t$ cannot be bounded in terms of $n$ and $s$ only.

In the next section we discuss how our results fit a broader landscape of existing Knapsack and Subset Sum algorithms.

\subsection{Related work}

\paragraph*{Pseudopolynomial time algorithms for Knapsack}

Bellman~\cite{bellman} was the first to show that the Knapsack problem admits a
pseudopolynomial time algorithm. He presented an $\Oh(Nt)$ time algorithm based on
dynamic programming.
Pisinger~\cite{pisinger91} gave an $\Oh(N s v)$ time algorithm, which is an improvement for
instances with both small sizes and small values.
He proved that only \emph{balanced feasible
solutions} to Knapsack need to be considered in order to find an optimal
solution. Then he used this observation to decrease the number of
states of the dynamic program. His arguments may be thought of as an early
example of proximity-based arguments.

Kellerer and Pferschy~\cite{kellerer04} studied approximation
algorithms for Knapsack. As a subroutine they developed a $\Ot(N+vp)$ time (exact) algorithm, where $p$ denotes the optimal total value. Their algorithm can be easily modified to work in $\Ot(N+st)$ time.
Their approach, based on fast $(\min,+)$-convolution for structured (convex)
instances was rediscovered and improved by Axiotis and Tzamos~\cite{AxiotisT19}.
Bateni et al.~\cite{bateni18} achieved the same $\Ot(N+st)$ running time with a different method, which can be seen as a far-reaching refinement of Pisinger's idea~\cite{pisinger91}.
They also developed the \emph{prediction} technique, which let them achieve $\Ot(N+vt)$ running time.

Eisenbrand and Weismantel~\cite{EisenbrandW18} studied more general integer linear programs, and presented a $\Ot(ns^2)$ time
algorithm for Knapsack with multiplicities, based on proximity-based arguments. To the best of our knowledge they
are the first to consider Knapsack with multiplicities.
Subsequently Axiotis and
Tzamos~\cite{AxiotisT19} improved logarithmic factors (in the non-multiplicity setting)
and gave an $\Oh(Ns^2)$ time algorithm, which they also generalized to $\Oh(Nv^2)$ time.
Bateni et al.~\cite{bateni18} also
explicitly consider the Knapsack problem with multiplicities and independently designed a $\Ot(ns^2 \min
\{n,s\})$ time algorithm. 

Axiotis and Tzamos suggested~\cite[Footnote~2]{AxiotisT19}
that the fast convex convolution can be combined
with proximity-based arguments of Eisenbrand and Weismantel~\cite{EisenbrandW18}
to obtain an algorithm for small items with running time independent of $t$.
However, a direct application of Eisenbrand and Weismantel~\cite{EisenbrandW18}
proximity argument (see \cite[Section 4.1]{EisenbrandW18}) reduces an instance to $t \le \Oh(ns^2)$, which, in
combination with $\Oh(N+st)$ algorithm, yields $\Oh(N+ns^3)$ runtime. Our
algorithm improves it to $\Oh(n+s^3)$ by a more careful proximity argument and
a convex convolution that explicitly handles negative items. Moreover, we
show how to extend this reasoning to the multiplicity setting.

\begin{remark}
    \label{rem:power}
    Lawler~\cite{lawler79} showed that the variant with multiplicities can be reduced to the
    0-1 variant. His reduction transforms a multiset composed of at most $u$
    copies of each of $n$ distinct numbers
    bounded by $s$ into an instance of $\Oh(n \log{u})$ numbers bounded by $\Oh(us)$. This
    easily enables us to adapt algorithms with no time dependence on $s$ (e.g.,
    the $\Oh(Nt)$ time algorithm of Bellman) into the setting with multiplicities (with
    logarithmic overhead).
\end{remark}

\begin{remark}
    \label{rem:harmon}
    There is also a folklore reduction that enables us to bound $N \le \Ot(t)$
    for the variant with multiplicities. For each $x \in [s]$ keep $\floor{t/x}$
    most profitable items of size $s_i = x$. This leaves us with at most $N \le
    \Oh(t \log(t))$ items and does not increase the item sizes.
\end{remark}

See Table~\ref{tab:history-knapsack} for a summary of the known results for Knapsack. 

\begin{table}
    \centering

    \caption{Pseudopolynomial time algorithms for \textbf{Knapsack}.
        $N$ is the total number of items,
        $n$ is the number of distinct items,
        $t$ is the knapsack capacity,
        $s$ is the maximum size and $v$ the maximum value of an item.
        Symbol (\---) means that no non-trivial optimization is given for the
        respective regime; running times can still be
        derived from the trivial inequalities $n \le N$ and $t \le Ns$. We use
        symbol $(\ddagger)$ when Remark~\ref{rem:power} applies and $(\dagger)$ when
        Remark~\ref{rem:harmon} applies.
        }
    \begin{tabular}{@{}lll@{}}

        \toprule
        \textbf{0-1 Knapsack} & \textbf{with multiplicities} & \textbf{Reference}   \\
        \midrule
         $\Oh(N t)$   &  $\Ot(n t)^\ddagger$ & Bellman~\cite{bellman} \\
         $\Oh(N s v)$ &  \--- & Pisinger~\cite{pisinger91} \\
         $\Ot(N + st)$ & $\Ot(n + st)^\dagger$  & Kellerer and Pferschy~\cite{kellerer04}, also~\cite{bateni18,AxiotisT19}\\
         $\Ot(N + vt)$ & $\Ot(n+ vt)^\dagger$  & Bateni et al.~\cite{bateni18}\\
         
          \---          & $\Ot(n s^2 \min \{ n,s\})$ & Bateni et al.~\cite{bateni18}\\
         $\Oh(N \min\{s^2,v^2\})$ & \--- & Axiotis and Tzamos~\cite{AxiotisT19}\\
          \--- & $\Ot(n s^2)$ & Eisenbrand and Weismantel~\cite{EisenbrandW18}\\
        \midrule
          \--- & $\Oh(n+\min\{s^3,v^3\})$ & \textbf{This paper}\\
        \bottomrule
    \end{tabular}
    \label{tab:history-knapsack}
\end{table}

\paragraph*{Pseudopolynomial time algorithms for Subset Sum}

Subset Sum is a special case of Knapsack and we expect
significantly faster algorithms for it. Pisinger's algorithm~\cite{pisinger91} runs in
$\Oh(Ns)$ time for Subset Sum. The first improvement in all parameter regimes
over the $\Oh(Nt)$ time algorithm of Bellman was given by Koiliaris and Xu~\cite{koiliaris-xu}. They
presented $\Ot(\sqrt{N}t + N)$, $\Ot(N + t^{5/4})$ and $\Ot(\Sigma)$ time
deterministic algorithms for Subset Sum, where $\Sigma$ is the total sum of
items. A by now standard method, used by all these algorithms, is to encode an instance of Subset Sum as a convolution problem that can be solved using Fast Fourier Transform.
Subsequently, Bringmann~\cite{Bringmann17} presented a $\Ot(N+t)$ randomized time algorithm for
Subset Sum based on the color-coding technique. Jin and
Wu~\cite{sosa-wu} later gave an alternative $\Ot(N+t)$ randomized time algorithm
based on Newton's iterative method. Their proof is notable for being very compact.

From a different perspective, Galil and Margalit~\cite{galil-margalit}
used additive combinatorics methods to prove that Subset Sum can be solved
in near linear time when $t \gg \Sigma s/N^2$ and all items are distinct.
Very recently~Bringmann and
Wellnitz~\cite{dense-subset-sum} generalized that result to multisets. Their algorithm
combined with the $\Ot(N+t)$ time algorithm~\cite{Bringmann17} yields a $\Ot(N+u^{1/2}s^{3/2})$ time algorithm for
Subset Sum with multiplicities (cf., Lemma~\ref{lem-addcomb}). For $u=1$ this gives 
the currently fastest $\Ot(N+s^{3/2})$ time algorithm (in terms of small $s$).
With our $\Ot(n+s^{5/3})$ time algorithm we improve upon their result for $u \gg s^{1/3}$.
We note that even with the naive (not binary) multiplicity encoding
our improvement is nontrivial, since the mentioned case with $u=1$ requires that each item has
a different size. For example, even an $\Oh(N + s^{1.99})$ time
algorithm does not follow immediately from~\cite{Bringmann17} when multiple items can have the same size.
We discuss their additive combinatorics methods in more detail in 
Section~\ref{sec:additive-comb-subset-sum-detail}.

See Table~\ref{tab:history-subsetsum} for a summary of the known results for Subset Sum.

\begin{table}
    \centering

    \caption{Pseudopolynomial time algorithms for \textbf{Subset Sum}.
        $N$ is the total number of items,
        $n$ is the number of distinct items,
        $t$ is the target value,
        $\Sigma$ is the sum of all items,
        $s$ is the maximum item, and $u$ is the maximum multiplicity of an item.
        Symbol (\---) means that no non-trivial optimization is given for the
        respective regime; running times can still be
        derived from the trivial inequalities $n \le N \le nu$ and $t \le Ns$.
        In $(\star)$ the instance cannot have two items with the same size,
        the algorithm works only for $u=1$.  We use
        symbol $(\ddagger)$ when Remark~\ref{rem:power} applies.
    }
        
    \begin{tabular}{@{}lll@{}}

        \toprule
        \textbf{0-1 Subset Sum} & \textbf{with Multiplicities} & \textbf{Reference} \\
        \midrule
        $\Oh(N t)$ & $\Ot(nt)^\ddagger$ & Bellman~\cite{bellman}   \\
        $\Oh(Ns)$ & \--- & Pisinger~\cite{pisinger91}\\
        $\Ot(N + \sqrt{N}t)$ & $\Ot(n + \sqrt{n}t)^\ddagger$ & Koiliaris and Xu~\cite{koiliaris-xu}\\
        $\Ot(N + t^{5/4})$ & $\Ot(n + t^{5/4})^\ddagger$ & Koiliaris and Xu~\cite{koiliaris-xu}\\
        $\Ot(\Sigma)$ & \--- & Koiliaris and Xu~\cite{koiliaris-xu}\\
        $\Ot(N + t)$ & $\Ot(n+t)^\ddagger$ & Bringmann~\cite{Bringmann17}\\
        $\Ot(N + s^{3/2})$ \hfill$\star$ & (not applicable) & Galil and Margalit~\cite{galil-margalit} \\
        \--- & $\Ot(N + u^{1/2}s^{3/2})$ & Bringmann and Wellnitz~\cite{dense-subset-sum} \\
        \midrule
        \--- & $\Ot(n + s^{5/3})$ & \textbf{This Paper}\\
        \bottomrule
    \end{tabular}
    \label{tab:history-subsetsum}
\end{table}

\paragraph*{Lower bounds}
Bringmann's Subset Sum algorithm~\cite{Bringmann17}, which runs in time $\Ot(N+t)$, was shown to be
near-optimal by using the modern toolset of fine-grained complexity. More
precisely, any $t^{1-\eps} 2^{o(N)}$ algorithm for Subset Sum, for any $\eps > 0$, would violate
both the Strong Exponential Time Hypothesis~\cite{subset-sum-lower2} and the Set Cover
Conjecture~\cite{subset-sum-lower}. This essentially settles the complexity of
the problem in the parameters $N$ and $t$.
These lower bounds use reductions that produce instances with $t=\Tht(s)$, and therefore they
do not exclude a possibility of
a $\Ot(N+s)$ time algorithm for Subset Sum. The question if such an algorithm
exists is still a major open problem~\cite{axiotis19}.

Bringmann and Wellnitz~\cite{dense-subset-sum}
excluded a possibility of a near-linear algorithm for Subset Sum in a \emph{dense} regime.
More precisely, they showed that, unless the Strong Exponential Time
Hypothesis and the Strong $k$-Sum Hypothesis both fail, Subset Sum requires $(s
\Sigma/(Nt))^{1-o(1)}$ time (where $\Sigma$ is the total sum of items).

For the Knapsack problem
Bellman's algorithm~\cite{bellman} remains optimal
for the most natural parametrization by $N$ and $t$.
This was explained by Cygan at el.~\cite{talg19} and Künnemann et al.~\cite{marvin17},
who proved an $(N+t)^{2-o(1)}$ lower bound assuming the $(\min,+)$-Convolution Conjecture.
Their hardness constructions create instances of
0-1 Knapsack where $s$ and $t$ are $\Tht(N)$. This is also the best lower bound
known for Knapsack with multiplicities. In particular, an $\Oh(N+s^{2-\eps})$ time algorithm
is unlikely, and our $\Oh(n+s^3)$ upper bound leaves the gap for the best exponent between $2$ and $3$.

\paragraph*{Other variants of Knapsack and Subset Sum}

We now briefly overview other variants of Knapsack and Subset Sum, which are not directly related to our results.
The Unbounded Knapsack problem is the special case with $u_i = \infty$, for all $i \in [n]$,
and one can assume w.l.o.g.~$N = n \le s$.
For that variant Tamir~\cite{tamir09} presented an $\Oh(n^2 s^2)$ time algorithm. 
Eisenbrand and Weismantel~\cite{EisenbrandW18}
improved this result and gave an $\Oh(ns^2)$ time algorithm using proximity arguments.
Bateni et al.~\cite{bateni18} presented a $\Ot(ns + s^2 \min\{n,s\})$ time algorithm.
Then, an $\Oh(n + \min\{s^2,v^2\})$ algorithm for Unbounded Knapsack was given
independently by Axiotis and Tzamos~\cite{AxiotisT19} and Jansen and Rohwedder~\cite{itcs19}.
Finally, Chan and He~\cite{DBLP:conf/esa/ChanH20} gave a $\Ot(ns)$ time algorithm.
Unbounded Knapsack seems to be an easier problem than 0-1 Knapsack because algorithms do not need to keep track of which items are already used in partial solutions. Most of the Unbounded Knapsack techniques do not apply to 0-1 Knapsack.

In the polynomial space setting, Lokshtanov and Nederlof~\cite{saving-space}
presented a $\Ot(N^4 s v)$ time algorithm for Knapsack
and a $\Ot(N^3 t)$ time algorithm for Subset Sum. The latter was subsequently
improved by Bringmann~\cite{Bringmann17}, who gave a $\Ot(Nt^{1+\eps)})$ time and
$\Ot(N\log{t})$ space algorithm. Recently, Jin, Vyas and
Williams~\cite{jin-soda21} presented a $\Ot(Nt)$ time and $\Ot(\log(Nt))$ space
algorithm (assuming a read-only access to $\Ot(\log{N}\log{\log{N}} +
\log{t})$ random bits).

In the Modular Subset Sum problem, all subset sums are
taken over a finite cyclic group $\mathbb{Z}_m$, for some given integer $m$.
Koiliaris~and~Xu~\cite{koiliaris-xu} gave a $\Ot(m^{5/4})$ time algorithm for
this problem, which was later improved by~\cite{axiotis19} to $\Oh(m
\log^7{m})$. Axiotis et al.~\cite{modular-subset-sum-21} independently
with Cardinal~and~Iacono~\cite{cardinal21} simplified
their algorithm and gave an $\Oh(m\log{m})$ time randomized and
$\Oh(m \polylog(m))$ deterministic time algorithms.
Recently, Pot\k{e}pa~\cite{potepa2020faster} gave the currently fastest $\Oh(m \log (m)
\alpha(m))$ deterministic algorithm for Modular Subset Sum (where $\alpha(m)$ is
the inverse Ackerman function).

Bringmann and Nakos~\cite{topkconv} designed a near-linear time algorithm for output
sensitive Subset Sum by using additive combinatorics methods. Finally, Jansen
and Rohwedder~\cite{itcs19} considered the Unbounded Subset Sum problem and
presented a $\Ot(s)$ time algorithm. 

\section{Techniques}

In this section we recall several known techniques from different fields,
which we later combine as black boxes in order to get efficient algorithms
in the setting with binary encoded multiplicities.
We do not expect the reader to be familiar with all of them,
and we include their brief descriptions for completeness.
Nevertheless, it should be possible to skip reading this section
and still get a high-level understanding of our results.

\subsection{Proximity arguments}
Now we introduce proximity arguments, which
will allow us to avoid a dependency on the multiplicities $u_i$ in the running time.
Very similar arguments were used by Eisenbrand and Weismantel~\cite{EisenbrandW18}
for more general integer linear programs. We reprove them for our simpler case to make the paper self-contained.

We will show that we can efficiently compute a solution to Knapsack which differs from an optimal solution only in a few items.
To this end, we define a \emph{\textbf{maximal prefix solution}} to be a solution obtained as follows.
We order the items by their efficiency, i.e.~by the ratios $v_i / s_i$, breaking ties arbitrarily.
Then, beginning with the most efficient item, we select items in the decreasing order of efficiency
until the point when adding the next item
would exceed the knapsack's capacity. At this point we stop and return the solution.

Note that a maximal prefix solution can be found in time $\Oh(n)$:
We first select the median of $v_i / s_i$ in time $\Oh(n)$~\cite{Blum73}.
Then, we check if the sum of all the items that are more efficient than the median exceeds $t$.
If so, we know that none of the other items are used in the maximal prefix
solution. Otherwise, all the more efficient items are used.
In both cases we can recurse on the remaining $n/2$ items. The running time is then of the
form of a geometric sequence that converges to $\Oh(n)$.

\begin{lemma}[cf., \cite{EisenbrandW18}]\label{lem:proximity}
  Let $p$ be a maximal prefix solution to Knapsack.
  There is an optimal solution $z$
  that satisfies $\lVert z - p \rVert_1 \le 2s$, where
  $\lVert z - p \rVert_1$ denotes $\sum_{i\in [n]} |z_i - p_i|$.
\end{lemma}
\begin{proof}
  Let $z$ be an optimal solution which minimizes $\lVert z - p \rVert_1$.
  If all the items fit into the knapsack, $p$ and $z$ must be equal.
  Otherwise, we can assume that the total sizes of both solutions,
  i.e.~$\sum_{i\in [n]} s_i z_i$ and $\sum_{i\in [n]} s_i p_i$,
  are both between $t - s + 1$ and $t$.
  In particular, we have that
  \begin{equation}\label{eq:proximity-diff}
      -s < \sum_{i\in [n]} s_i (p_i - z_i) < s.
  \end{equation}

  For the sake of the proof consider the following process.
  We start with the vector $z - p$, and we move its components towards zeros,
  carefully maintaining the bounds of~(\ref{eq:proximity-diff}).
  That is, in each step of the process, if the current sum of item sizes is positive, we reduce a positive component by $1$;
  if the sum is negative, we increase a negative component by $1$.
  The crucial idea is that during this process in no two steps we can have the same
  sum of item sizes. Otherwise, one could apply to $z$ the additions and removals performed
  between the two steps, and therefore obtain
  another solution that is closer to $p$ and is still optimal.
  Indeed, the optimiality follows from the fact that this operation does not increase the total size of the solution,
  and it also cannot decrease the value of the solution, because every item selected
  by $p$ but not by $z$ has efficiency no lower than every item selected by $z$ but not by $p$.
  Hence, the number of steps, i.e., $\lVert z - p \rVert_1$, is bounded by $2s$.
\end{proof}
This lemma can be used to avoid the dependency on multiplicities $u_1,u_2,\dotsc,u_n$ as follows.
We compute a maximal prefix solution $p$.
Then we know that there is an optimal solution $z$ with
\begin{equation*}
  z_i \in \{0, \dotsc, u_i\} \cap \{p_i - 2s, \dotsc, p_i + 2s\}
  \quad\forall i\in [n] .
\end{equation*}
Hence, we can obtain an equivalent instance by fixing the choice of some items.
Formally, we remove $\max\{0, p_i - 2s\}$ many copies of item $i$ and subtract
their total size from $t$.
If some item still has more than $4s$ copies, we can safely remove the excess.
This shows that one can reduce a general knapsack instance to an instance
with $u_i \le 4s$ for all $i\in [n]$.
In particular, a naive application of Bellman's algorithm would run in time
$\Oh(t\cdot\sum_{i=1}^n u_i) \le \Oh(n^2 s^3)$.
Later in this paper we will apply the same proximity statement in more involved arguments.

\subsection{Fast structured $(\min,+)$-convolution}

Another technique that we use in this paper is a fast algorithm for
structured instances of the $(\min,+)$-convolution problem. This technique was
already applied in several algorithms for
Knapsack~\cite{kellerer04,bateni18,AxiotisT19}.
We use it to find solutions for all knapsack capacities in $\{0,\dotsc,t\}$
in time $\Oh(n + st + t\log^2(t))$. We consider a slightly more general variant, where the values of items
may also be negative and the knapsack constraint has to be satisfied with equality.
To avoid confusion, we let $\bar v_i\in\mathbb Z$, $i\in [n]$, denote these
possibly negative values of items.
\begin{lemma}\label{lem:maxplus}
  Let $\bar v_1,\bar v_2,\dotsc,\bar v_n \in \mathbb Z$.
  In time $\Oh(n + st + t \log^2(t))$ one can solve
{\rm
\begin{align}
  &\max_{x\in\mathbb Z^n} \sum_{i \in [n]} \bar v_i x_i &&\text{ subject to }
  \sum_{i \in [n]} s_i x_i = t' &&\text{ and } \quad \forall_{i\in [n]}\ 0 \le x_i \le u_i,
\label{eq:eqknapsack}
\end{align}
}
for all $t' \in \{0,\dotsc,t\}$.
\end{lemma}
\begin{proof}
For $h\in [s+1]$ let $w^{(<h)} = \langle w^{(<h)}_0, w^{(<h)}_1,\dotsc,w^{(<h)}_t \rangle$
where $w^{(<h)}_{t'}$ denotes the value of an optimal solution to~\eqref{eq:eqknapsack} for the knapsack capacity $t'$
when restricting the instance to items $i$ with $s_i < h$.
If there is no solution satisfying the equality constraint with $t'$, we let $w^{(<h)}_{t'} = - \infty$.
Our goal is to compute the vectors $w^{(<1)}, w^{(<2)},\dotsc, w^{(<h+1)}$ iteratively.
We define the vector
$w^{(h)} = \langle w^{(h)}_0, w^{(h)}_1, \dotsc, w^{(h)}_t \rangle$
which describes the optimal solutions solely of items $j$ with $s_j = h$.
For each $i$, the component $w^{(h)}_{t'}$ for $t'=ih$ is equal
to the total value of the $i$ most valuable items of size $h$, or to $-\infty$ if there are less than $i$
items of size $h$. All components for indices not divisible by $h$ are $-\infty$.

Hence, to compute $w^{(h)}$ it suffices to find the $\ceil{t/h}$ most valuable items of size $h$
in the decreasing order of values.
In time $\Oh(n + s)$ we partition the items by their size $s_i$. Extracting
the $t / h$ most valuable items for all $h \in [s]$ requires in total a time of
$\Oh(n + \sum_{h\in [s]} t/h) \le \Oh(n + t \log(t))$.
Finally, sorting all sets takes in total
$\Oh(\sum_{h \in [s]} t/h \cdot \log(t/h)) \le \Oh(t \log^2(t))$ time.

Given $w^{(<h)}$ for some $h$ we want to compute the vector $w^{(<h+1)}$ in
time $\Oh(t)$. Then the lemma follows by iteratively applying this step. To this end
we notice that $w^{(<h+1)}$ is precisely the $(\max,+)$-convolution
of $w^{(h)}$ and $w^{(<h)}$, that is,
\begin{equation*}
  w^{(<h+1)}_i = \max\left\{ w^{(h)}_{j} + w^{(<h)}_{i - j}  \mid j\in\{0,\dotsc,i\} \right\} .
\end{equation*}
While in general computing a $(\max,+)$-convolution is conjectured to require
quadratic time~\cite{talg19,marvin17}, in this case it
can be done efficiently by exploiting the simple structure of $w^{(h)}$.
For each remainder $r \in \{0,\dotsc,h-1\}$ we separately compute the entries of indices
that are equal to~$r$~modulo~$h$.
We define the matrix $M\in \mathbb Z^{\lceil t/h \rceil \times \lceil t/h \rceil}$ with
\begin{equation*}
   M[i,j] = w^{(<h)}_{jh + r} + w^{(h)}_{(i - j) h},
\end{equation*}
where $w^{(h)}_{(i - j) h} = -\infty$ if $j > i$.
We do not explicitly construct the matrix, but we can compute any entry of $M$ in
the constant time.
To produce the vector $w^{(<h+1)}$ it suffices to find the maximum of each row of $M$.
This can be done efficiently, since $M$ is \emph{inverse-Monge}, that is,
\begin{multline*}
  M[i, j] + M[i + 1, j + 1]
  = w^{(<h)}_{jh + r} + w^{(h)}_{(i - j)h} + w^{(<h)}_{jh + h + r} + w^{(h)}_{(i - j)h} \\
  \ge w^{(<h)}_{jh + r} + w^{(h)}_{(i - j)h + h} + w^{(<h)}_{jh + h + r} + w^{(h)}_{(i - j)h - h}
  = M[i + 1, j] + M[i, j + 1] .
\end{multline*}

Therefore, we can compute the row maxima in time $\Oh(t/h)$ with SMAWK
algorithm~\cite{DBLP:journals/algorithmica/AggarwalKMSW87}.  This implies a
total running time of $\Oh(t)$ for all remainders $r$ and proves the lemma.
\end{proof}

\subsection{Additive combinatorics}
\label{sec:additive-comb-subset-sum}

In this section, we introduce a near-linear time algorithm
for dense instances of Subset Sum, more precisely, instances with $N^2 \gg u s$.

The techniques behind the algorithm were introduced by Galil~and~Margalit~\cite{galil-margalit}, and 
recently generalized to the multiset setting by Bringmann~and~Wellnitz~\cite{dense-subset-sum}.
We focus on the modern description of~\cite{dense-subset-sum}, and show
in Section~\ref{sec:recovering-results}
that in our application we can additionally report a solution.

\begin{theorem}[Bringmann and Wellnitz~\cite{dense-subset-sum}]
  \label{prop-addcomb}
  There exists $\lambda = \Tht(u s \Sigma / N^2)$ such that in time $\Ot(N)$ we can
  construct a data structure that for any $t$ satisfying $\lambda \le t \le
  \Sigma/2$ decides in time $\Oh(1)$ whether $t$ is a subset sum.
\end{theorem}

This is non-trivial when $\lambda \le \Sigma/2$, that is when $N^2 \gg us$.
We note that in the setting of Subset Sum with multiplicities
Theorem~\ref{prop-addcomb} gives an $\Ot(N + u^{1/2}s^{3/2})$ time algorithm.
We denote by $\Ss(I)$ all subset sums of a multiset $I$, and we write
$\Sigma(I) = \sum_{a\in I} a$.

\begin{lemma}
        
    \label{lem-addcomb}

    Given a multiset $I$ of size $N$, in $\Ot(N + s^{3/2} u^{1/2})$ time we can construct a data
    structure that, for any $t \in \nat$,
    \begin{enumerate}[(a)]
      \item \label{lem-addcomb-decision} determines whether $t \in \Ss(I)$ in time $\Oh(1)$, and
      \item if $t \in \Ss(I)$, it finds $X \subseteq I$ with $\Sigma(X) = t$ in time $\Ot(N + s^{3/2} u^{1/2})$.
    \end{enumerate}
\end{lemma}
In order to prove that we can retrieve a solution with the given running time
we will need to get into technical details behind the proof of Theorem~\ref{prop-addcomb}.
We do it in Section~\ref{sec:recovering-results}.
Now, we sketch how to use Theorem~\ref{prop-addcomb}
as a blackbox to give an $\Ot(N + u^{1/2}s^{3/2})$ time algorithm that can
only detect if there is a solution.

\begin{proof}[Proof of Lemma~\ref{lem-addcomb}~(\ref{lem-addcomb-decision})]
    
    Let $\lambda$ be defined as in Theorem~\ref{prop-addcomb}.
    If the total sum of items $\Sigma$ is bounded by $\Ot(s^{3/2}u^{1/2})$,
    then we can use Bringmann's $\Ot(N+t)$ time Subset Sum algorithm~\cite{Bringmann17}
    to compute $\Ss(I)$. Therefore from now on we can assume that $s^{3/2}u^{1/2}
    \le \Ot(\Sigma) \le \Ot(Ns)$. In particular, this means that $\sqrt{us} \le \Ot(N)$. Hence,
    \begin{displaymath}
      \lambda \le \Ot\left(\frac{us \Sigma}{N^2}\right) \le
      \Ot\left(\frac{us^2}{N}\right) \le \Ot(u^{1/2}s^{3/2})
      .
    \end{displaymath}
    This means that we can afford $\Ot(\lambda)$ time. 
    In time $\Ot(N + \lambda)$ we find all subset sums in $\Ss(I) \cap [0,\lambda]$
    using Bringmann's algorithm~\cite{Bringmann17}.
    For $t \in [\lambda, \Sigma/2]$ we use Theorem~\ref{prop-addcomb} to decide
    in $\Ot(N)$ time if $t \in \Ss(I)$. For $t > \Sigma / 2$ we ask about $\Sigma - t$ instead.
\end{proof}

\section{Knapsack with small item sizes}
\label{sec:knapsack}

In this section we obtain an $\Oh(n + s^3)$ time algorithm for Knapsack
by combining the proximity and convolution techniques.

\begin{proof}[Proof of Theorem~\ref{thm:knapsack}]
Let $p$ be a maximal prefix solution. By Lemma~\ref{lem:proximity}
there is an optimal solution $z$ with $\lVert z - p \rVert_1 \le 2s$.
We will construct an optimal solution $x$ that is composed of three parts, that is,
$x = p - x^- + x^+$.
Our intuition is that $x^+$ is supposed to mimic the items that are included in $z$
but not in $p$. We denote these items by
$(z - p)_+$, where $(\cdot)_+$ takes for each component the maximum of it and $0$.
Likewise, $x^-$ intuitively stands for $(p - z)_+$, the items in $p$ but not in $z$.

To find $x^+$ and $x^-$ we will invoke twice the $\Oh(n + st + t \log^2(t))$ time algorithm of Lemma~\ref{lem:maxplus}.
Let $\Delta = t - \sum_{i \in [n]} s_i p_i$, that is, the remaining knapsack capacity in the prefix solution. We can assume w.l.o.g.\ that $\Delta < s$, since otherwise $p$ already includes all items and must be optimal.
We use Lemma~\ref{lem:maxplus} to compute optimal solutions to the following integer programs for every $k \in \{0,\dotsc,2s^2 + \Delta\}$.

\begin{align}
  &\max_{x\in \mathbb Z^n} \sum_{i \in [n]} v_i x_i & \text{ subject to } &&
  \sum_{i \in [n]} s_j x_i \le k & \quad \text{and} \quad \forall_{i\in [n]}\; 0 \le x_i \le u_i - p_i
  \label{eq:solplus} \\
  &\max_{x\in \mathbb Z^n} \sum_{i \in [n]} -v_i x_i & \text{ subject to } &&
  \sum_{i \in [n]} s_i x_i = k & \quad \text{and} \quad \forall_{i\in [n]}\; 0 \le x_i \le p_i
  \label{eq:solminus}
\end{align}

We denote the resulting solutions by $x^+(k)$ and $x^-(k)$.
Note, that formally the algorithm in Lemma~\ref{lem:maxplus} outputs solutions to the
variant of~\eqref{eq:solplus} with equality, we can transform it 
to the above form with a single pass over the solutions.

For any $k$ the solution $x(k) = p - x^-(k) + x^+(k + \Delta)$
is feasible. We compute values of all such solutions and select the best of them.
To show that this is indeed an optimal solution, it suffices to show that
for one such $k$ the solution is optimal.
Let $k = \sum_{i \in [n]} \max\{0, s_i (p_i - z_i)\} \le 2 s^2$, then
$(x - z)_+$ is feasible
for~(\ref{eq:solplus}) with $k+\Delta$ and $(p - z)_+$ for~(\ref{eq:solminus}) with $k$.
Thus,
\begin{align*}
  \sum_{i\in [n]} v_i (x(k))_i &= \sum_{i \in [n]} v_i p_i &-& \sum_{i \in [n]} v_i (x^-(k))_i &&+ \sum_{i\in [n]} v_i (x^+(k + \Delta))_i \\
  &\ge \sum_{i \in [n]} v_i p_i &-& \sum_{i \in [n]} v_i\max\{0, p_i - z_i\} &&+ \sum_{i \in [n]} v_i\max\{0, z_i - p_i\}
  = \sum_{i \in [n]} v_i z_i .
\end{align*}
It remains to bound the running time. The maximal prefix solution can be found in
time $\Oh(n)$. Each of the two calls to the algorithm of Lemma~\ref{lem:maxplus} takes time
$\Oh(n + s^3 + s^2 \log^2(s)) = \Oh(n + s^3)$ and selecting the best solution among
the $2s^2$ candidates takes time $\Oh(s^2)$.
\end{proof}

\section{Knapsack with small item values}
\label{sec:small-v}
In this section we show that it is also possible to solve Knapsack in time $\Oh(n + v^3)$, proving Theorem~\ref{thm:knapsack-small-values}.
This can be derived directly from the $\Oh(n+s^3)$ time algorithm from the previous section.
Essentially, we swap the item values and sizes by considering the complementary problem of finding the items that are not taken in the solution.
Then our goal is to solve
\begin{align}\label{eq:ks-complement}
  &\min_x \sum_{i\in [n]} v_i x_i &&\text{ subject to }  \sum_{i\in [n]} s_i x_i \ge \sum_{i\in [n]} u_i s_i - t &&\text{ and } \quad \forall_{i\in [n]} 0 \le x_i \le u_i .
\end{align}
Suppose we are satisfied with any solution that has value at least some given
$v^\star$.
Then this can be solved by
\begin{align*}
  &\max_x \sum_{i\in [n]} s_i x_i &&\text{ subject to } \sum_{i\in [n]} v_i x_i
  \le \sum_{i\in [n]} u_i v_i - v^\star &&\text{ and } \quad \forall_{i\in [n]} 0 \le x_i \le u_i .
\end{align*}
Notice that this is now a Knapsack problem with item sizes bounded by $v$.
Hence, our previous algorithm can solve it in time $\Oh(n + v^3)$.
It remains to find the optimum of~(\ref{eq:ks-complement}) and use it for the value of $v^\star$.
Notice that the maximal prefix solution $p$ gives a good estimate of this
$v^\star$, because its value
is between $v^\star - v + 1$ and $v^\star$.
Thus, one could in a straight-forward way implement a binary search for $v^\star$ and this would increase
the running time only by a factor of $\log(v)$, but we can avoid this and get an $\Oh(n+v^3)$ time algorithm.

It is enough to devise an algorithm that in time $\Oh(n + v^3)$ computes a
solution for each of the $v$ potential values values of $v^\star$ at once. Then we
can return the largest $v^\star$ for which the solution requires a knapsack of size
at most $t$. Fortunately, our original Theorem~\ref{thm:knapsack} can compute
solution to every $t' \in \{t-v,t-v+1,\ldots,t\}$ and the $\Oh(n+v^3)$ time algorithm for Knapsack follows. 

We include a small modification of Knapsack algorithm from
Section~\ref{sec:knapsack} for completeness.

\begin{claim}
    In $\Oh(n+s^3)$ time we can compute an optimal solution to Knapsack for
    every $t' \in \{t-s,t-s+1,\ldots,t\}$.
\end{claim}
\begin{proof}
    Recall that the intermediate solutions $x^+(k)$ and $x^-(k)$ depend only on the maximal
    prefix solution $p$ and the only property of $p$ that is needed is that it differs from
    the optimal solution by $\Oh(s)$ items.
    Notice that the maximal prefix solutions with respect to each of the values $t'$ above
    differ only by at most $s$ items. Hence, $p$, the prefix solution for $t$,
    differs from each of the optimal solutions only by $\Oh(s)$.
    Hence, we only need to compute $x^+(k)$ and $x^-(k)$ once.
    Given these solutions the remaining computation takes only $\Oh(s^2)$ for each $t'$;
    thus, $\Oh(s^3)$ in total.
\end{proof}

\section{Subset Sum}

In this section we give a $\Ot(n+s^{5/3})$ time algorithm for Subset Sum with multiplicties proving Theorem~\ref{thm:subsetsum}.
Our algorithm is a combination of additive combinatorics and proximity arguments.
Throughout this section, we denote by $\Ss(A)$ the set of all subset sums of
a (multi-)set of integers $A$. In other words, $t \in \Ss(A)$ if there exists some $B\subseteq A$
with $\Sigma(B) = t$.

\begin{algorithm}[ht!]
  \SetAlgoLined
  \DontPrintSemicolon
  \SetKwInOut{Input}{Algorithm}
  \SetKwInOut{Output}{Output}
    \Input{$\mathtt{SubsetSum}(I,t)$. }
    \Output{Multiset $X \subseteq I$ with $\Sigma(X) = t$ or $\mathtt{NO}$ if such a multiset does not exists. }
    Preprocess $I$ using Lemma~\ref{lem:proximity} so that $u_i \le O(s)$ for all $i\in [n]$\\
    Set $k := \floor{s^{1/3}}$\\
    Construct $I^{\up} := \{ (\max \{0,\floor{u_i/k} - 8\}, s_i) \; : \; (u_i,s_i) \in I\}$ \\
    Construct $I^{\dn} := \{ (u_i - k \cdot \max \{0,\floor{u_i/k} - 8\}, s_i) \; : \; (u_i,s_i) \in I\}$ \\
    Construct oracle to $\Ss(I^\dn)$ \tcp*{with Lemma~\ref{lem-addcomb}}
    Construct set of candidates $\Cc(I^\up) \subseteq \Ss(I^\up)$ \tcp*{with Lemma~\ref{lem:candidates}}
    \ForEach{$t' \in \Cc(I^\up)$}{
        \If{$t - k \cdot t' \in \Ss(I^{\dn})$}{
            Recover $A \subseteq I^\up$ with $\Sigma(A) = t'$\\
            Recover $B \subseteq I^\dn$ with $\Sigma(B) = t - k \cdot t'$\\
            \Return $(k \cdot A) \cup B$ 
        }
    }
    \Return $\mathtt{NO}$
    \caption{$\Ot(n+s^{5/3})$ time algorithm for Subset Sum with multiplicties}
  \label{alg:subset-sum}
\end{algorithm}

We now give a high level overview of the algorithm (see Algorithm~\ref{alg:subset-sum}).
In the following we assume w.l.o.g.\ that $u\le
\Oh(s)$ by using the preprocessing described after Lemma~\ref{lem:proximity}.
First, we split the instance $I$ into two parts $I^\up$ and $I^\dn$:
Let $k = \lfloor s^{1/3} \rfloor$.
For every item $s_i$ with multiplicity $u_i$
we add  $u^\up_i =  \max \{0,\floor{u_i/k} - 8\}= \Oh(s^{2/3})$ many items of size $s_i$ into multiset
$I^\up$. The rest of the items of size $s_i$, i.e., $u^\dn_i = u_i - k \cdot
u^\up_i = \Oh(s^{1/3})$ many, are added to multiset $I^\dn$. 
Intuitively, $I^\up$ stands for taking bundles of $k$ items. The
set $I^\dn$ consists of the remaining items. In particular, it holds that:
\begin{equation*}
    \Ss(I) = \left\{ k t^\up + t^\dn \; : \; t^\up \in  \Ss(I^\up) \text{ and } t^\dn \in \Ss(I^\dn)\right\} .
\end{equation*}
Our goal is to decide whether there exists an integer $t' \in \Ss(I^\up)$ with the
property that $t - kt' \in \Ss(I^\dn)$. The strategy of the algorithm is as
follows: we will use proximity arguments to bound the number of candidates
for such a $t' \in \Ss(I^\up)$ and efficiently enumerate them in
time $\Ot(s^{5/3})$. 

\begin{lemma}
    \label{lem:candidates}
    In time $\Ot(s^{5/3})$, we can construct $\Cc(I^\up) \subseteq \Ss(I^\up)$
    of size $|\Cc(I^\up)| \le \Ot(s^{5/3})$ with the property that if $t \in
    \Ss(I)$, then there exists $t' \in \Cc(I^\up)$ and $t - kt' \in \Ss(I^\dn)$.
    Moreover, for any $t' \in \Cc(I^\up)$ we can find $X \subseteq I^\up$ with
    $\Sigma(X) = t'$ in time $\Ot(s^{5/3})$.
\end{lemma}
We will prove this lemma in Section~\ref{sec:proximity}.
To check the condition $t - kt' \in \Ss(I^\dn)$ we observe
that the set $I^\dn$ has bounded multiplicity and large density. To accomplish
that we use Lemma~\ref{lem-addcomb}. It uses a
recent result of Bringmann and Wellnitz~\cite{dense-subset-sum} and enables us
to decide in constant time if $t - kt' \in \Ss(I^\dn)$ for any $t'$ after
a preprocessing that requires time $\Ot(|I^{\dn}| + s^{3/2}(u^\dn)^{1/2}) \le
\Ot(s^{5/3})$ because $|I^{\dn}| \le \Oh(s \cdot u^\dn)$ and $u^\dn \le
\Oh(s^{1/3})$. We extend their methods to be able to construct the solution within $\Ot(s^{5/3})$
time (see Section~\ref{sec:recovering-results} for the proof of Lemma~\ref{lem-addcomb}).

Now, we analyse the correctness of Algorithm~\ref{alg:subset-sum}.
The running
time is bounded by $\Ot(n + s^{5/3})$ because the number of candidates is $|\Cc(I^\up)| \le
\Ot(s^{5/3})$ by Lemma~\ref{lem:candidates}. The algorithm is correct
because set $\Cc(I^\up)$ has the property that if an answer to the Subset Sum is
positive, then there exists $t' \in \Cc(I^\up)$ with $t - k \cdot t' \in \Ss(I^{\dn})$

Moreover, when the answer is positive
we have $t' \in \Cc(I^\up)$ with $t - k \cdot t' \in \Ss(I^\dn)$.
We use Lemma~\ref{lem:candidates} to recover $A \subseteq I^\up$ with
$\Sigma(A) = t'$. Then we use Lemma~\ref{lem-addcomb} to find $B \subseteq
I^\dn$ with $\Sigma(B) = t - k \cdot t'$. We construct a final solution $(k \cdot A) \cup B$ by
unbundling items in $A$ (duplicating them $k$ times) and joining them with set $B$.
This concludes the proof of Theorem~\ref{thm:subsetsum}.

\subsection{Finding a small set of candidates}
\label{sec:proximity}
In this section we derive a small set of candidates $\Cc(I^\up)$ and prove
Lemma~\ref{lem:candidates}. This is based on the proximity result (Lemma~\ref{lem:proximity}).
Recall that for $p$, a maximal prefix solution, we know that there exists some
feasible solution that differs only by $\Oh(s)$ items (if the instance is feasible).
Suppose we split $p$ between $I^\up$ and $I^\dn$ into $p^\up$ and $p^\dn$.
As each of the items in $I^\up$ stands for $k$ items in $I$, one might expect
that the part of the optimal solution that comes from $I^\up$ differs from $p^\up$
by only $\Oh(s / k)$ items. This is not necessarily true if $p^\up$ and $p^\dn$ are
chosen unfavorably. Fortunately, we can show
that with a careful choice of $p^\up$ and $p^\dn$ it can be guaranteed.

\begin{definition}[Robust split]
    \label{bundle-preserving}
    Let $p = (p_1,\ldots,p_n) \in \nat^n$ be a maximal prefix solution.
    Let $p^\up,p^\dn \in \nat^n$ be defined by 
    \begin{displaymath}
        p^\up_i = \begin{cases} 
            0 &\text{if } p_i \le 4 k \text{ or } u_i \le 8k, \\
            \floor{u_i/k} - 8 &\text{if }  p_i \ge u_i - 4 k \text{ and } u_i > 8 k, \\
            \floor{p_i/k} - 2 &\text{if }  p_i \in (4k, u_i - 4 k) \text{ and } u_i > 8 k,
        \end{cases}
    \end{displaymath}
    and $p^\dn_i = p_i - k p^\up_i$, for every $i \in [n]$.
    We call $p^\up, p^\dn$ the \emph{robust split} of $p$.
\end{definition}
An important property of the choice of $p^\up$ and $p^\dn$ is that there is
some slack for $p^\dn_i$. Namely, if we were to change $p_i$ slightly (say, by less than $k$)
then we only need to
change $p^\dn_i$ (and do not need to change $p^\up_i$) to maintain $p_i = k p^\up_i + p^\dn_i$.

\begin{lemma}
  \label{lem:prox-bundle}
  Let $p$ be a maximal prefix solution and let $p^\up,p^\dn$ be its robust split.
  If $t \in \Ss(I)$, then there are solutions $x^\up,x^\dn$ of $I^\up$
  and $I^\dn$ such that:
  \begin{displaymath}
      \sum_{i\in [n]} (kx^\up_i + x^\dn_i) s_i = t  \;\; \text{ and } \;\; \norm{ x^\up - p^\up }_1 \le \Oh(s/k).
  \end{displaymath}
\end{lemma}

\begin{figure}
    \centering
    \includegraphics[width=0.7\textwidth]{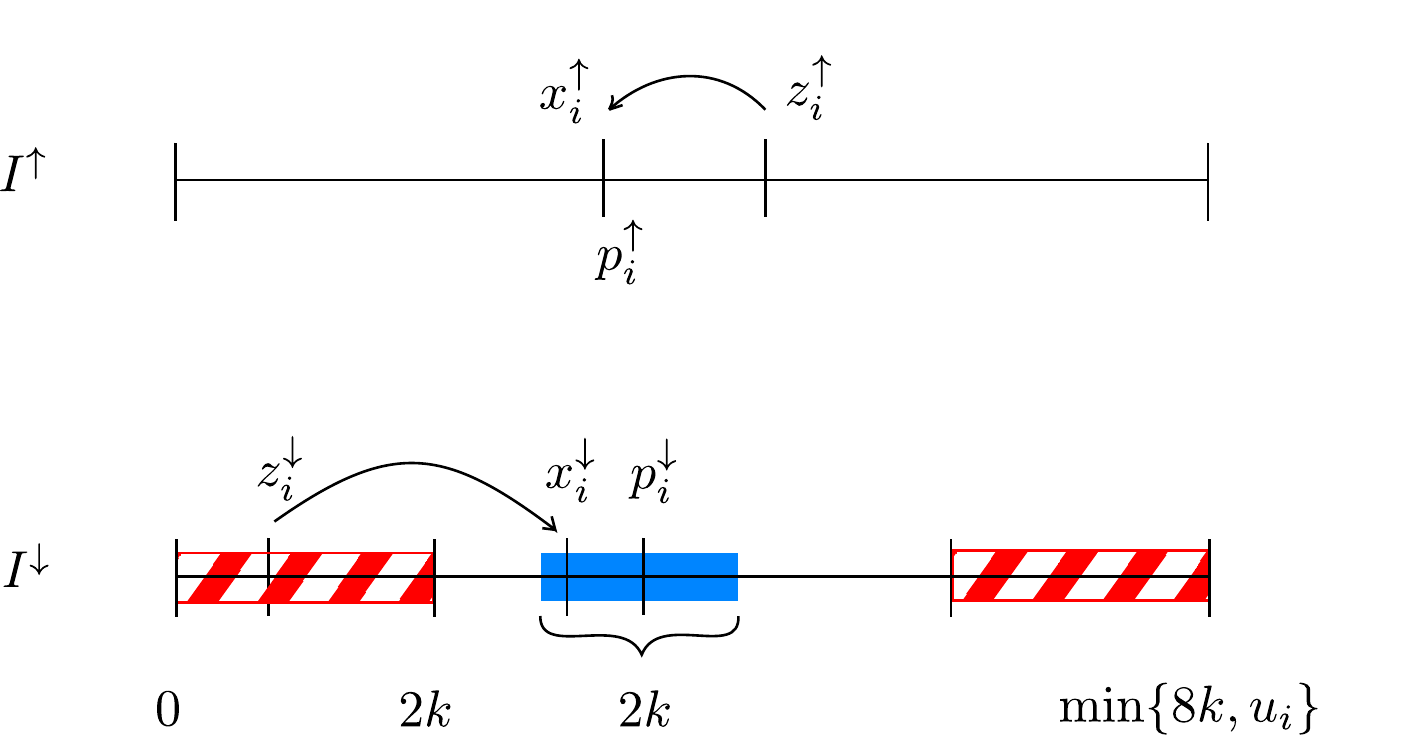}
    \caption{Schematic idea behind the proof of Lemma~\ref{lem:prox-bundle} in
    the case $u_i > 8k$.  We select $p^\dn_i$ such that it is never in the red
    hatched regions. The property of $p^\up,p^\dn$ is that any solution that differs by at most $k$
    elements can take the same elements from $I^\up$ as $p^\up_i$. 
    In that case the situation where the optimal solution is of the form $z$ (in the figure)
    can always be avoided to instead get the situation with $x$.}
    \label{fig:proximity}
\end{figure}

\begin{proof}
    The proof consists of straightforward, but tedious, calculations; see Figure~\ref{fig:proximity} for an intuition.

    By Lemma~\ref{lem:proximity} we know that there is a feasible solution $x$
    with $\norm{x - p}_1 \le \Oh(s)$. Let us use this solution and construct
    $x^\up$ and $x^\dn$. We consider each index $i \in [n]$ individually and show that
    there is a split of $x_i$ into $x^\up_i$, $x^\dn_i$ which are feasible for
    $I^\up$, $I^\dn$, that is, they satisfy the bounds $[0, u^\up_i]$ and $[0, u^\dn_i]$,
    and we have that $|x^\up_i - p^\up_i| \le 19 |x_i - p_i| / k$. This implies the lemma.
    To this end, consider two cases based on $|x_i -p_i|$.
    
    \paragraph*{Case 1: $|x_i - p_i| < k$}

    In this case we set $x^\up_i = p^\up_i$ and $x^\dn_i = x_i - k x^\up_i$. This
    choice  does not contribute anything to the norm
    $\norm{x^\up - p^\up}_1$ because $x_i^\up = p_i^\up$. We need to show that
    $x^\up_i$ and $x^\dn_i$ are feasible
    solutions to the subset sum instances $I^\up$ and $I^\dn$. 
    
    Clearly, $x^\up_i$ is between $0$ and $u^\up_i = \max\{\lfloor u_i / k \rfloor - 8, 0\}$
    in the first two cases of Definition~\ref{bundle-preserving}.
    In the last case, we have that $4k < p_i < u_i - 4k$ and therefore
    $0 \le \lfloor p_i / k \rfloor - 2 = p^\up_i$.
    Furthermore, $p^\up_i = \lfloor p_i / k \rfloor - 2 < p_i / k - 3 < u_i / k - 7 < \lfloor u_i / k \rfloor - 8 = u^\up_i$.

    Therefore, it remains to show that $x^\dn_i$ is feasible for $I^\dn$.
    More precisely, we will prove that:
    \begin{displaymath}
        0 \le x_i - k p^\up_i \le u_i - k \cdot \max\{ \floor{u_i/k}-8,0\}.
    \end{displaymath}
    To achieve that, we will crucially rely on our choice for $p^\up$ in Definition~\ref{bundle-preserving}. 

    \paragraph*{Case 1a: $u_i \le 8k$ or $p_i \le 8k$}
    When $u_i \le 8 k$ or $p_i \le 8 k$ then the claim follows because $p_i^\up = x_i^\up = 0$. 
    
    \paragraph*{Case 1b: $p_i \ge u_i - 4k$ and $u_i > 8k$}
    In this case we have $p^\up_i := \floor{u_i/k} - 8$.
    Inequality $x_i - k p^{\up}_i \ge 0$ follows from $x_i > p_i - k \ge u_i - 5k \ge k \lfloor u_i / k \rfloor - 8k = k p^\up_i$.
    Next, we use the fact $x_i \le u_i$ to conclude $x_i - k \floor{u_i/k} +8k \le u_i - k \floor{u_i/k} + 8k$.  
    
    \paragraph*{Case 1c: $4k < p_i < u_i - 4k$ and $u_i > 8k$}
    
    In this case we have $p^\up_i := \floor{p_i/k}-2$.
    The inequality $x_i > p_i -k \ge k \floor{p_i/k} - k = k p^\up_i + k$ shows that $x_i -
    kp^\up_i \ge 0$.
    Next, we use inequality $p_i - k \floor{p_i/k} < k$ to show that
    \begin{equation*}
        x_i - k \floor{p_i/k} \le p_i - k \floor{p_i/k} + k\le 2k < \underbrace{u_i - k \floor{u_i/k}}_{\ge 0} + 8k .
    \end{equation*}
    
    \paragraph*{Case 2: $|x_i - p_i| \ge k$}
    We set $x^\up_i := \min\{\floor{x_i/k}, u^\up_i\}$ and
    $x^\dn_i := x_i - k x^\up_i$.
    Clearly, $0 \le x^\up_i \le u^\up_i$. Furthermore,
    $x^\dn_i \ge 0$ and if $x^\up_i = u^\up_i$ then $x^\dn_i \le u_i - k u^\up_i = u^\dn_i$;
    otherwise, $x^\dn_i = x_i - k\floor{x_i / k} < k < u^\dn_i$.

    It remains to bound the difference between
    $x^\up_i$ and $p^\up_i$.
    Note that because of our choice of $p^\up$ we have $|p^\up_i - \floor{p_i/k}| \le 8$.
    Also, $|u^\up_i - \floor{u_i / k}| \le 8$.
    Therefore,
    \begin{multline*}
        |x^\up_i - p^\up_i| \le \big|\min\{\floor{x_i/k},u^\up_i\} - \floor{p_i/k}\big| + 8
        \le \big|\underbrace{\min\{\floor{x_i/k},\floor{u_i/k}\}}_{= \floor{x_i / k}} - \floor{p_i/k}\big| + 16 \\
        \le |x_i - p_i| / k + 18 .
    \end{multline*}
    Recall that we assumed $|x_i - p_i| \ge k$. Thus,
    \begin{equation*}
      |x^\up_i - p^\up_i| \le 19 |x_i - p_i| / k . \qedhere
    \end{equation*}
\end{proof}

Now we are ready to proceed with the algorithmic part and prove Lemma~\ref{lem:candidates}.

\begin{proof}[Proof of Lemma~\ref{lem:candidates}]
    Let $t_p$ be the value of $p^\up$ in $I^\up$, that is, $t_p :=
    \sum_{i \in [n]} p^\up_i s_i$.
    Let $A^-$ be the multiset of numbers selected in $p^\up$.
    Moreover, let $A^+$ denote all other elements in $I^\up$.
    We now compute
    \begin{displaymath}
        S^+ := \Ss(A^+) \cap [c \cdot s^{5/3}]  \;\; \text{ and } \;\; S^- :=
        \Ss(A^-) \cap [c \cdot s^{5/3}] .
    \end{displaymath}  
    Here $c$ is a constant that we will specify later.
    These sets can be computed in time $\Ot(s^{5/3})$ with Bringmann's
    algorithm~\cite{Bringmann17}. Next, using FFT we compute the sumset
    \begin{equation*}
        \Cc(I^\up) := \{t_p + a - b \mid a \in S^+, b\in S^-\} .
    \end{equation*}

    Any element in $\Cc(I^\up)$ is an integer of the form $t_p + a -b$,
    integer $a$ is the sum of elements not in $p^\up$, and integer $t_p - b$ is the contribution
    of elements in $p^\up$.
    This operation
    takes time $\Ot(s^{5/3})$ because the range of values of $S^+$ and $S^-$
    is bounded by $\Oh(s^{5/3})$. We return set $\Cc(I^\up)$ as the set of possible
    values of the candidates. To recover a solution we will use the fact that
    Bringmann's algorithm can recover solutions and the property that we can find
    a witness to FFT computation in linear time.
    Since $S^+$ and $S^-$ are subsets of $[c \cdot s^{5/3}]$ we know that
    $\Cc(I^\up)$ is a subset of $\{t_p - c \cdot s^{5/3}, \dotsc, t_p + c\cdot s^{5/3}\}$.
    In particular, $|\Cc(I^\up)| \le \Ot(s^{5/3})$.

    It remains to to show that $\Cc(I^\up)$ contains some $t'$ such that $t - k t'\in \Ss(I^\dn)$
    if $t\in \Ss(I)$.
    By Lemma~\ref{lem:prox-bundle} it suffices to show that $\Cc(I^\up)$ contains
    all values $\sum_{i\in [n]} x_i^\up s_i$ for $x^\up$ with $\lVert x^\up - p^\up\rVert_1 \le c \cdot s/k$,
    where $c$ is the constant in the lemma.
    This holds because $S^+$ contains $\sum_{i\in [n] : x^\up_i \ge p^\up_i} (x_i^\up - p^\up_i) s_i$
    and $S^-$ contains $\sum_{i\in [n] : x^\up_i \le p^\up_i} (p^\up_i - x_i^\up) s_i$.
\end{proof}

\subsection{Introduction to additive combinatorics methods}
\label{sec:additive-comb-subset-sum-detail}

In this section we review the structural ideas behind the proof of
Theorem~\ref{prop-addcomb}. Next, in Section~\ref{sec:recovering-results} we show how to
use them to recover a solution to Subset Sum with multiplicities.

The additive combinatorics structure that we
explore is present in the regime when $N^2 \gg us$. We formalize this
assumption as follows:

\begin{definition}[Density]
    We say that a multiset $X$ is $\delta$-dense if it satisfies $|X|^2 \ge
    \delta u s$.
\end{definition}
Note that if all numbers in $X$ are divisible by the same integer $d$,
then the solutions to Subset Sum are divisible by $d$. Intuitively, this
situation is undesirable, because our goal is to exploit the density of the
instance. With the next definition we quantify how close we are to the case
where almost all numbers in $X$ are divisible by the same number.

\begin{definition}[Almost Divisor]
    We write $X(d) := X \cap d \mathbb{Z}$ to denote the multiset of all
    numbers in $X$ that are divisible by $d$ and $\ol{X(d)} := X \setminus X(d)$
    to denote the multiset of all numbers in $X$ not divisible by $d$. We say
    that an integer $d > 1$ is an $\alpha$-almost divisor of $X$ if $|\ol{X(d)}|
    \le \alpha u \Sigma(X)/|X|^2$.
\end{definition}
Bringmann and Wellnitz~\cite{dense-subset-sum} show that this situation is not the hardest case. 
\begin{lemma}[Algorithmic Part of~\cite{dense-subset-sum}]
    \label{lem:reduce-divisor-free}
    Given an $\Tht(1)$-dense multiset $X$ of size $N$ in time $\Ot(N)$ we
    can compute an integer $d \ge 1$, such that $X' := X(d)/d$ is
    $\Tht(1)$-dense and has no $\Tht(1)$-almost divisors.
\end{lemma}
They achieve that with novel prime factorization techniques. This lemma allows them
essentially to reduce to the case that there are no $\Tht(1)$-almost divisors.
We will give more details on this in the end of Section~\ref{sec:additive-comb-subset-sum}.
In our proofs we use
the Lemma~\ref{lem:reduce-divisor-free} as a blackbox. Next, we focus on the
structural part of their arguments.

\paragraph*{Structural part}

The structural part of~\cite{dense-subset-sum} states the surprising property. If we
are given a sufficiently dense instance with no almost divisors then every set
with target within the given region is attainable.

\begin{theorem}[Structural Part of~\cite{dense-subset-sum}]
    \label{thm:structural-part}
    If $X$ is $\Tht(1)$-dense and has no $\Tht(1)$-almost divisor then
    $[\lambda_X,\ldots,\Sigma(X)-\lambda_X ] \subseteq \mathcal{S}(X)$ for some
    $\lambda_X = \Tht(u s \Sigma(X)/|X|^2))$.
\end{theorem}
Therefore,~Bringmann and Wellnitz~\cite{dense-subset-sum} after the reduction to the
almost-divisor-free setting can simply output \texttt{YES} on every target in the
selected region. Our goal is to recover the solution to subset sum and
therefore, we need to get into the details of this proof and show that we can
efficiently construct it. The crucial insight into this theorem is the following
decomposition of a dense multiset.

\begin{lemma}[Decomposition, see~{\cite[Theorem 4.35]{dense-subset-sum}}]
    \label{lem:decomposition}
    Let $X$ be a $\Tht(1)$-dense multiset of size $n$ that has no
    $\Tht(1)$-almost divisor. Then there exists a partition $X = R \uplus A
    \uplus G$ and an integer $\kappa = \Ot(u \Sigma(X)/|X|^2)$ such that:

    \begin{itemize}
        \item set $\mathcal{S}(R)$ is $\kappa$-complete, i.e., $\mathcal{S}(R) \bmod \kappa = \mathbb{Z}_\kappa$,
        \item set $\mathcal{S}(A)$ contains an arithmetic progression
            $\mathcal{P}$ of length $2s$ and step size $\kappa$ satisfying
            $\max\{\Pp\} \le \Ot(u s \Sigma(X)/|X|^2)$,
        \item the multiset $G$ has sum $\Sigma(G) \ge \Sigma(X)/2$.
    \end{itemize}
\end{lemma}
Now, we sketch the proof of Theorem~\ref{thm:structural-part} with the
decomposition from Lemma~\ref{lem:decomposition}. This is based on the proof
in~\cite{dense-subset-sum}.

\begin{proof}[Sketch of the proof of Theorem~\ref{thm:structural-part} assuming Lemma~\ref{lem:decomposition}]
    We show that any target $t \in [\lambda_X,\ldots,\Sigma(X)-\lambda_X]$ is a
    subset sum of $X$.
    For that we will assume without loss of generality that $t \le \Sigma(X) / 2$
    (note that $t$ is a subset sum if and only if $\Sigma(X) - t$ is).
    By Lemma~\ref{lem:decomposition} we get a partition of
    $X$ into $R \uplus A \uplus G$. We know that $\mathcal{S}(A)$ contains an arithmetic
    progression $\Pp \subseteq \mathcal S(A)$, with $\Pp =
    \{a+\kappa,a+2\kappa,\ldots,a+2s\kappa\}$. We
    construct a subset of $X$ that sums to $t$ as follows. First, we greedily
    pick $G' \subseteq G$ by iteratively adding elements until:
    \begin{displaymath}
       t - \Sigma(G') \in [a+\kappa(s+1), a+\kappa (s+1) + s].
    \end{displaymath}
    This is possible because the largest element is bounded by $s$, $t$ is at most
    $\Sigma(X)/2 \le \Sigma(G)$, and $\lambda_X$ is selected such that:
    \begin{displaymath}
        t \ge \lambda_X \ge a+\kappa(s+1).
    \end{displaymath}
    The next step is to select a subset $R' \subseteq R$ that sums up to a number congruent to
    $(t-\Sigma(G') - a)$ modulo $\kappa$. Recall, that set $R$ is
    $\kappa$-complete, hence such a set must exist. Moreover, w.l.o.g.~$R' < \kappa$, and $\Sigma(R') <
    \kappa s$.
    Therefore, we need extra elements of total sum
    \begin{displaymath}
        t - \Sigma(G'\cup R') \in [a+\kappa, a+2\kappa s],
        \quad \text{and} \quad
        t - \Sigma(G'\cup R') \equiv a \mod \kappa.
    \end{displaymath}
    Finally, we note that this is exactly the range of elements of the
    arithmetic progression $\Pp$. It means that we can pick a subset $A' \subseteq
    A$, that gives the appropriate element of the arithmetic progression $\Pp$
    and $t = \Sigma(G'\cup R' \cup A')$.
\end{proof}

\subsection{Recovering a solution}
\label{sec:recovering-results}

In this section, we show how to recover a solution to Subset Sum.  We
need to overcome several technical difficulties. First, we need to reanalyze
Lemma~\ref{lem:decomposition} and show that the partition $X = A \uplus R \uplus
G$ can be constructed efficiently. This step follows directly
from~\cite{dense-subset-sum}.  However, we do not know of an efficient way to
construct $\kappa$ and $a$. We show that we do not really need it. Intuitively,
for our application we can afford to spend a time $\Ot(N + \lambda_I)$.
This observation enables us to use the $\Ot(N+t)$ time
algorithm of Bringmann~\cite{Bringmann17} to reconstruct the solution. We
commence with the observation that the decomposition into $R \uplus A \uplus G$ can be
constructed within the desired time.

\begin{claim}[Recovering decomposition]
    \label{lem:decomposition-construct}
    Let $X$ be a $\Tht(1)$-dense multiset that has no
    $\Tht(1)$-almost divisor. Let $K = 42480 \cdot u \cdot \Sigma(X)
    \log(2u)/|X|^2$.  Then in $\Ot(N + u \Sigma(X)/|X|^2)$ time we can explicitly find
    a partition $X = R \uplus A \uplus G$ such that:
    \begin{itemize}
        \item set $\mathcal{S}(R)$ is $d$-complete for any $d \le K$,
        \item there exists an integer $\kappa \le K$ such that the set $\mathcal{S}(A)$ contains an arithmetic progression
            $\mathcal{P}$ of length $2s$ and step size $\kappa$ satisfying
            $\max\{\Pp\} \le \Ot(u s \Sigma(X)/|X|^2)$,
        \item the multiset $G$ has sum $\Sigma(G) \ge \Sigma(X)/2$.
    \end{itemize}
\end{claim}

\begin{proof}[Proof of Claim~\ref{lem:decomposition-construct} based on arguments from~\cite{dense-subset-sum}]
    We follow the proof of Theorem 4.35 in \cite{dense-subset-sum}.
    We focus on the construction of partition $X$ and present
    only how the construction follows from~\cite{dense-subset-sum}.

    To construct the set $R$, Bringmann and Wellnitz use
    \cite[Theorem 4.20]{dense-subset-sum}. We start
    by picking an arbitrary subset $R' \subseteq X$ of size $\tau = \Tht(u
    \Sigma(X)/|X|^2)$. Next we generate the set $S$ of all prime numbers $p$ with $p \le \tau$. We
    can do this in $\Ot(\tau)$ time by the sieve of Eratosthenes
    algorithm~\cite{sorenson1990introduction}. Then, we compute the prime
    factorization of every number in $R'$ by \cite[Theorem 3.8]{dense-subset-sum}
    in $\Ot(\tau)$ time. This enables us to
    construct the set $P$ of primes $p$ with $p \le \tau$ such that every $p \in P$
    does not divide at least $\tau/2$ numbers in $R'$. Bringmann~and~Wellnitz~\cite{dense-subset-sum} show that
    $|P| \le 2 \log{s}$. Next, for every $p \in P$ we select an arbitrary subset
    $R_p \subseteq \overline{X(p)}$ of size $|R_p| = \tau$. This can be done in $\Ot(N)$ time because $|P| =
    \Ot(1)$. Finally~Bringmann and Wellnitz~\cite{dense-subset-sum} construct $R
    = R' \cup \left(\bigcup_{p
    \in P} R_p\right)$. See Theorem 4.20 in~\cite{dense-subset-sum} for a
    proof that the constructed $R$ is $d$-complete for any $d \le K$.

    To construct set $A$ we do exactly the same as Bringmann and Wellnitz
    \cite{dense-subset-sum} and we pick at most $\floor{n/4}$ smallest elements
    from $X \setminus R$. At the end we set $G = X \setminus (R \cup A)$.
    See~\cite{dense-subset-sum} for a proof that $A$ and $G$ have a desired
    properties.
\end{proof}

Now, we are ready to prove our result about recovering the solution

\begin{lemma}[Restatement of Lemma~\ref{lem-addcomb}]
    Given a multiset $I$ of $N$ elements, in $\Ot(N + s^{3/2} u^{1/2})$ time we can construct a data
    structure that for any $t \in \nat$ can decide if $t \in \Ss(I)$ in time
    $\Oh(1)$. Moreover if $t \in \Ss(I)$ in $\Ot(N + s^{3/2} u^{1/2})$ time we can
    find $X \subseteq I$ with $\Sigma(X) = t$.
\end{lemma}

In the rest of this section we prove Lemma~\ref{lem-addcomb}.  We assume that all considered $t$ are at most
$\Sigma(I)/2$ because for larger $t$ we can ask about $\Sigma(I) - t$ instead.
Let $\lambda_I$ and be defined as in Theorem~\ref{prop-addcomb}.
When the total sum of the elements $\Sigma(I)$ is bounded by
$\Ot(s^{3/2}u^{1/2})$, Bringmann's algorithm~\cite{Bringmann17}
computes $\Ss(I)$ and retrieves
the solution in the declared time. Therefore, we can assume that $s^{3/2}u^{1/2}
\le \Ot(|I|s)$. In particular, this means that $\sqrt{us} \le \Ot(|I|)$ and the
set $I$ is $\Tht(1)$-dense. Hence,
\begin{displaymath}
  \lambda_I \le \Ot\left(\frac{us \Sigma(I)}{|I|^2}\right) \le
  \Ot\left(\frac{us^2}{|I|}\right) \le \Ot(u^{1/2}s^{3/2})
  .
\end{displaymath}
This means, that we can afford $\Ot(\lambda_I)$ time.  In time $\Ot(|I| +
\lambda_I)$ we can find all subset sums in $\Ss(I)$ that are smaller than
$\Tht(\lambda_I)$ using Bringmann's algorithm. Additionally, when $t \le
\Tht(\lambda_I)$ Bringmann's algorithm can find $X \subseteq I$ in the output
sensitive time.

We are left with answering queries about targets greater than $\lambda_I$. To
achieve that within $\Ot(\lambda_I)$-time preprocessing we use the Additive Combinatorics result
from~\cite{dense-subset-sum}.

Observe, that if we were interested in a data structure that works in $\Ot(N + \lambda_I) \le \Ot(N + s^{3/2} u^{1/2})$
and \emph{decides} if $t \in \Ss(I)$ we can directly use \cite{dense-subset-sum} as a
blackbox. Therefore, from now, we show that we can also reconstruct the
solution in $\Ot(N + \lambda_I)$ time.

\begin{claim}
   We can find a set $X \subseteq I$ with $\Sigma(X) = t$ in
   $\Ot(N + \lambda_I)$ time for any $t \in [\lambda_I, \Sigma(I)/2]$ if such an $X$ exists.
\end{claim}

\begin{proof}
  First, we use Lemma~\ref{lem:reduce-divisor-free} to find an integer
  $d \ge 1$, such that set $I' := I(d)/d$ has no $\Tht(1)$-almost divisor
  and is $\Tht(1)$-dense. Observe, that the integer $d$ can be found in
  $\Ot(N)$ time and set $I'$ can be constructed in $\Ot(N)$ time.

  Bringmann and Wellnitz~\cite[Theorem 3.5]{dense-subset-sum} prove that (recall that $I$ is $\Tht(1)$-dense and $t \ge \lambda_I$):
  \begin{displaymath}
      t \in \Ss(I) \text{ if and only if } t \bmod d \in (\Ss(I) \bmod d).
  \end{displaymath}
  Therefore, the first step of our algorithm is to use Axiotis et
  al.~\cite{axiotis19} to decide if $t \bmod d \in (\Ss(I) \bmod d)$ and recover
  set $D \subseteq I \setminus I(d)$ if such a set exists (Axiotis et
  al.~\cite{axiotis19} enables to recover solution and by density assumption it
  works in $\Ot(N)$ time).
  
  If such a set exists by reasoning in~\cite{dense-subset-sum} we know that
  a solution must exist (otherwise we output \texttt{NO}). We are left with recovering set $K \subseteq I'$ with
  $\Sigma(K) := (t - \Sigma(D))/d$.

  Next, we use Claim~\ref{lem:decomposition-construct} to find a partition
  $I' = R \uplus A \uplus G$. We can achieve that in $\Ot(u \Sigma(I')/|I'|^2)
  \ll \Ot(N + \lambda_I)$ time. Ideally, we would like to repeat the reasoning presented
  in the proof of Theorem~\ref{thm:structural-part}. Unfortunately, we do
  not know how to explicitly construct $a$ and $\kappa$ within the given time.
  Nevertheless, Theorem~\ref{thm:structural-part} guarantees that
  $t \in \Ss(I)$ and that such integers $a$ and $\kappa$ exist.

  Based on the properties of sets in
  Theorem~\ref{thm:structural-part} we have
  \begin{displaymath}
      t \in \Ss(G \uplus R \uplus A) = \Ss(G) \oplus (\Ss(R \cup A) \cap [0,\Ot(\lambda_I)]).
  \end{displaymath}
  With a Bringmann's algorithm we can compute the set $T := \Ss(R \cup A) \cap
  [0,\Ot(\lambda_I)]$ in $\Ot(N + \lambda_I)$ time.
  Now, recall that in the proof of Theorem~\ref{thm:structural-part} we
  have chosen set $G'
  \subseteq G$ greedily to satisfy $\Sigma(G') \ge t- a -\kappa
  (s+1)$ for some $a,\kappa$ and $s$. Therefore it is enough to
  iterate over every greedily chosen $G' \subseteq G$ and
  check whether $t - \Sigma(G') \in T$. Notice that there are only $N$ options for $G'$
  that can be generated in $\Oh(N + \lambda_I)$ time.
  For each of them we can check whether $t - \Sigma(G') \in T$ in $\Oh(1)$ time
  because we have access to $T$. If we find such a set, we just report $K := G' \cup
  T'$, where $T' \subseteq R \cup A$ with $\Sigma(T') + \Sigma(G') = t$.
\end{proof}
This concludes the proof of Lemma~\ref{lem-addcomb}.

\bibliographystyle{plainurl}
\bibliography{bib}

\end{document}